\documentclass[12pt,twoside]{article}

\usepackage{amsmath,amsfonts,amssymb,amsthm,graphicx}     %formules mathématiques
\usepackage{color}

\newcommand{\R}{\mathbb{R}}

\newcommand{\E}{\mathbb{E}}

\newcommand{\Z}{\mathbb{Z}}

\newcommand{\Pro}{\mathbf{P}}
\newcommand{\Qb}{\mathbf{Q}}
\newcommand{\NN}{\mathbf{N}}

\newtheorem{prop}{Proposition}[section]
\newtheorem{Proposition}{Proposition}[section]

\newtheorem{definition}{Definition}

\def\b{\mathbb }

\def\rr{\b R}
\def\nn{\b N}

\def\zz{\b Z}

\def\11{{\bf 1}}

\def\ce{{\cal E}}

\def\cf{{\cal F}}

\def \E{ \mathbb E}

\def \R{\mathbb R}

\def \1{\mathbf{1}}

\begin{document}

\begin{center}
{\Large{\textbf{\hspace{-.9cm} \textrm{Jump-diffusion modeling in emission markets}}}}
\end{center}

\begin{center}
\textsc{K. Borovkov\footnote{Department of Mathematics and Statistics, University
Melbourne, Parkville 3010, Australia.}, G. Decrouez$^1$ and J. Hinz\footnote{National
University of Singapore, Department of Mathematics, 2 Science Drive, 117543
Singapore.}}
\end{center}

\renewcommand{\thefootnote}{}
\footnotetext{Research supported by the ARC Discovery Grant DP088069, the Start-up grant R-146-000-107-133 of the National University of Singapore and the research grants WBS R-703-000-020-720
/ C703000 of the Risk Management Institute at the National University of Singapore.}

%\author{Kostya Borovkov}
%\ead{borovkov@unimelb.edu.au}
%\address{Department of Mathematics and Statistics, The University of Melbourne, Parkville, Australia}
%\author{Geoffrey Decrouez}
%\ead{dgg@unimelb.edu.au}
%\address{Department of Mathematics and Statistics, The University of Melbourne, Parkville, Australia}
%\author{Juri Hinz \corref{cor1}}
%\ead{mathj@nus.edu.sg}
%\address{National University of Singapore,
%Department of Mathematics,
% 2 Science  Drive, 117543 Singapore}
%\cortext[cor1]{Corresponding author: Juri Hinz}
%\tnotetext[t1]{Research supported by the ARC Discovery Grant DP088069, the Start-up grant R-146-000-107-133 of the National University of Singapore and the research grants WBS R-703-000-020-720
%/ C703000 of the Risk Management Institute at the National University of Singapore.}

\begin{abstract}
Mandatory emission trading schemes are being established around the world.
Participants of such market schemes are always exposed to risks. This leads to the
creation of an accompanying market for emission-linked derivatives. To evaluate the
fair prices of such financial products, one needs appropriate models for the evolution
of the underlying assets, emission allowance certificates. In this paper, we discuss
continuous time diffusion and jump-diffusion models, the latter enabling one to model
information shocks that cause jumps in allowance prices. We show that the  resulting
martingale dynamics can be described in terms of non-linear partial differential and
integro-differential equations and use  a finite difference method to investigate
numerical  properties of their discretizations. The results are illustrated by a small numerical study.%

\bigskip


\medskip

\noindent {\em ASM Subject Classifications:} Primary: 91B70; Secondary: 91B60, 91B76,
93E20.

\end{abstract}

\section{Introduction}

Emission Trading Schemes (ETSs) have recently been designed with the main aim to reduce emissions of greenhouse gases  and other pollutants. Two examples of ETSs are the EU ETS (European Union Emission Trading Scheme) and the US Sulfur Dioxide Trading System. In such schemes, the regulator allocates each market participant a number of credits, each of which gives the holder the right to emit a specified amount of pollutant (e.g.\ CO$_2$). At the end of a compliance period, each agent should not have released more pollutant than covered by the credits it holds at that time, or it will have to pay a fine proportional to the excess of the pollutant amount over the one corresponding to the credits held. During the compliance period, however, each agent can trade credits with other market participants, depending on whether it is cheaper to reduce emission or to buy credits.

The present study relies on a detailed mathematical model for such a scheme based on market equilibrium. We exploit its implications for the risk neutral allowance price evolution.
The key issue is a  feedback relationship between allowance prices  and  pollution abatement activity. Namely,
any increase in allowance
price  enforces emission reduction since agents would then tend to  sell their allowances. Hence the increase of the  allowance price encourages supply of certificates  and lowers the probability of non-compliance, which,
in its turn, tends to diminish allowance  prices. Following \cite{HinzNovikov}, we describe
this feedback relationship in terms of  market fundamentals, using a stylized
evolution of the expected  non-compliance  and  emission savings costs.
We pay particular attention to discontinuities in the information flow, assuming
that some events may cause market participants to  change their views on the future reduction volumes
required to reach compliance. We show how
risk neutral  allowance price reacts to such  ``information discontinuity" and  study this effect in
derivatives valuation.

The paper is organized as follows. In Section \ref{Model}, we extend the discrete-time framework from
\cite{HinzNovikov} to the continuous-time setting. In Section \ref{ContinuousTime}, we interpret our results  in the context of diffusion processes, derive the corresponding partial differential equations and discuss  option pricing.
Section \ref{Jumps} is devoted to modeling shock  events  in the information flow. Our description is based on
jump-diffusion processes  and requires solving a nonlinear partial integro-differential equation. We show that solutions to these non-linear equations from Section 3 and 4 satisfy the maximum principle. In the last section we propose a numerical implementation of its solution. In particular, we prove that the discretized problem possesses a unique solution and satisfies the maximum principle.

The suggested allowance  price model is suitable for option pricing. It turns out that the fair price dynamics of European options written on emission allowance prices can be obtained by solving a linear partial integro-differential equation. We illustrate our results by numerical examples.

The literature on this subject is rather extensive, and we refer the interested reader to
a nice expository work \cite{Taschini2009} which reviews the fundamental concepts of the environmental economics
and provides a valuable guide to publications, which, however, is far from being complete.
The {\it economic theory} of allowance trading  goes back to \cite{Dales} and \cite{Montgomery}, where the authors proposed trading  the public good {\it environment} by means of transferable permits.
Important results in {\it dynamic allowance trading} were obtained in  \cite{Cronshaw}, \cite{Tietenberg2}, \cite{Rubin}, \cite{Leiby}, \cite{Schennach}, \cite{StevensRose}, \cite{Maeda} and in the literature cited therein.
Recently, after the introduction of the real-world emission market EU ETS, the
{\it empirical evidence} has become available. The experience gained from the operation of Phase I of the EU ETS is discussed in \cite{DaskalakisPsychoyiosMarkellos}, and a detailed  analysis of spot and futures allowance prices
from  this market is given   in \cite{UhrigHomburgWagner1} and
\cite{UhrigHomburgWagner2}. %There,  the demand for derivative instruments
%in emission markets is  also addressed.
The contributions \cite{BenzTrueck} and \cite{PaolellaTaschini} are devoted to {\it econometric modeling}
of emission allowance prices. Beyond confirming stylized facts of financial time
 series for prices of emission allowances,    Markov switch  and AR-GARCH
models are suggested. The modeling of
{\it dynamic price equilibrium} is addressed in
\cite{CarmonaFehrHinz} and \cite{CarmonaFehrHinzPorchet}, which provide a  mathematical analysis of the market equilibrium
and use optimal stochastic control theory to show  social optimality of emission trading schemes.
A recent work  \cite{HinzNovikov} considers equilibrium of risk averse market players and elaborates on
 risk neutral dynamics. The problems of
{\it derivatives valuation} in emission markets are also addressed.  The paper \cite{ChesneyTaschini} discusses an endogenous emission
permit price dynamics within equilibrium setting and elaborates on the valuation of  European options on emission allowances.
The dissertation \cite{Wagner} and the  paper \cite{Seifert} deal with the risk-neutral allowance price formation
within the EU ETS. In that paper, when utilizing  equilibrium properties, the  price evolution is
treated in terms of  marginal abatement costs and optimal stochastic control.
The work \cite{CetinVerschuere} is also devoted to option pricing within EU ETS.
The authors suppose that the drift of allowance spot prices is related to a hidden variable,
which describes the overall market position in allowance contracts, and make use of
filtering techniques to derive option price formulas which reflect specific allowance
banking regulations valid in the EU ETS. Finally,  the recent work \cite{CarmonaHinz} presents an  approach
where emission certificate futures are modeled in terms of a
deterministic time change applied to a certain class of  interval-valued diffusion processes.

\section{Modeling emission markets in discrete time}
\label{Model}

During the compliance period (typically 3--5 years), each ETS participant dynamically adjusts its
production processes (and hence its emissions) and trades
emission credits at times $t=0,1,2,\hdots,T<\infty$ to maximize its revenue. In this setting,
allowance price reaches  its equilibrium determined by supply and demand of emission allowances. In what follows,
we base on the model from \cite{HinzNovikov} which characterizes  the equilibrium allowance prices in terms of non-compliance uncertainty and abatement costs. This characterization forms the starting point for our analysis.

Let $\left(\Omega, \cf, \Pro, \{\cf_{t}\}_{t=0}^{T}\right)$ be  a filtered probability space. We assume that all the processes considered in this section are adapted to $\{\cf_{t}\}_{t=0}^{T}$.
Consider a market with a finite set $I$ of agents who must comply with the ETS rules.
Assume that, for  each
$i \in I$, an exogenously given  stochastic process $\{E^{i}_{t}\}_{t=0}^{T-1}$ describes the so-called
``business as usual" emission of agent $i$. That is,
$E^{i}_{t}= E^{i}_{t}(\omega)$ stands for the total pollution of the agent $i$ which will be emitted during the time
interval $(t, t+1]$ if no abatement measures are applied by the agent. Suppose that each agent $i$
can decide, at any time $t=0,  \dots, T-1$, to perform  a reduction of  $\xi^{i}_{t}$  pollutant units to be
emitted during  $(t, t+1]$.
 The cost of abatement is modeled by a function of the reduced volume that
can be random (reflecting the uncertainty in fuel  prices). Thus, if at time $t=0, \dots, T-1$ agent $i$ decides on reduction  by  $x \in [0, \infty)$  units for the time interval $(t,t+1]$, then
it costs it $C^{i}_{t}(x)=C^{i}_{t}(x)(\omega)$, where $C^{i}_{t} : [0, \infty) \times \Omega \mapsto [0,
\infty)$ is ${\cal B}([0, \infty))\otimes\cf_{t}$-measurable, and for each $\omega \in \Omega$,
the mapping $x \mapsto C^{i}_{t} (x)(\omega)$  is strictly convex and continuous with $C_{t}(0)=0$.
Therefore, following this abatement policy  $\{\xi^{i}_{t}\}_{t=0}^{T-1}$, agent $i$ will have born
by the compliance date $T$ the total abatement costs of
\begin{equation}
\sum_{t=0}^{T-1}C^{i}_{t}(\xi^{i}_{t}). \label{abatement}
\end{equation}
For each $i  \in I$, $\omega \in \Omega$, $t=0, \dots, T-1$ and $a \in [0, \infty)$, we introduce
the ``reduction volume"
\begin{equation} \label{volumefunction}
r_{t}^{i}(a)=r_{t}^{i}(a)(\omega):={\rm argmax}\left\{ax-C_{t}^{i}(x)(\omega) \, : \, x \in [0, E^{i}_{t}(\omega)]\right\}.
\end{equation}
This quantity gives the ``locally optimal" reduction volume for agent $i$ for the time period $(t,t+1]$ given that the price of one allowance unit is equal to $a$ for that time period; we may  assume (and  Proposition  \ref{mainprop} (b) below supports this) that, being rational, agent $i$ will implement emission reduction at that level. As is well known (see e.g. \cite{HinzNovikov}), under the above assumptions $r_{t}^{i}(a)(\omega)$ is a non-decreasing and continuous function of $a\in[0,\infty)$ for each $\omega \in \Omega$ and $t=0, \dots, T-1$.
Next denote the total $t$-th time period reduction by
 \begin{equation} \label{cfunt}
  r_{t}(a):=\sum_{i \in I}r^{i}_{t}(a), \qquad a \in [0, \infty),
 \end{equation}
which represents the total reduction by all agents in the market for the time period $(t,t+1]$, given the time $t$ price of one allowance unit is $a$.

Suppose that, at any time $t=0, \dots, T$, credits can be traded at the  spot price $A_{t}$.
Denote by $\vartheta^{i}_{t}$  the change in the allowance number held by agent $i$ at time $t$. Then, given the allowance prices $\{A_{t}\}_{t=0}^{T}$, the position changes $\{\vartheta^{i}_{t}\}_{t=0}^{T}$ will result for agent $i$ in the total trading costs of
\begin{equation}
\sum_{t=0}^{T}\vartheta^{i}_{t}A_{t}. \label{trading}
\end{equation}
Here and in what follows, for simplicity's sake we assume, zero interest rates (or, equivalently, that all the prices are already discounted).

Further, the total pollution of agent $i$ during the compliance interval $[0,T]$ can be expressed as the cumulative business-as-usual emission less the agent's total reduction:
$$
\sum_{t=0}^{T-1}E^{i}_{t} - \sum_{t=0}^{T-1}\xi^{i}_{t}.
$$
Denoting by $\gamma_i$ the agent's initial allowance allocation, we observe that $i$ will hold $\gamma_i+\sum_{t=0}^{T}\vartheta^{i}_{t}$ allowances at time $T$ and hence its loss resulting from the potential penalty payment at rate $\pi$ (the penalty for the amount of emissions corresponding to one allowance) is
\begin{eqnarray}
\pi\left[\sum_{t=0}^{T-1}(E^{i}_{t}-\xi^{i}_{t}-\vartheta^{i}_{t})-\gamma^{i}- \vartheta^{i}_{T}\right]^{+}. \label{penaltypayment}
\end{eqnarray}

Finally, we define the space of feasible trading strategies $\vartheta^{i}=\{\vartheta^{i}_{t}\}_{t=0}^{T}$
 and abatement strategies $\xi^{i}=\{\xi^{i}_{t}\}_{t=0}^{T-1}$ of agent $i \in I$ as
\begin{eqnarray*}
{\cal U}^{i}:=\{ (\vartheta^{i}, \xi^{i}) \, : \,
0 \le \xi^{i}_{t} \le E^{i}_{t}, \quad t =0, \dots, T-1\}. \label{admiss}
\end{eqnarray*}
In view of (\ref{abatement}), (\ref{trading}) and (\ref{penaltypayment}), the total revenue of agent $i$ following an admissible policy $(\vartheta^{i}, \xi^{i}) \in {\cal U}^{i}$ is equal to
\begin{equation*}  \label{totrev}
L^{A,i}(\vartheta^{i}, \xi^{i}):=-\sum_{t=0}^{T-1}(\vartheta^{i}_{t}A_{t}+ C^{i}(\xi^{i}_{t})) -\vartheta^{i}_{T}A_{T}
- \pi\left[\sum_{t=0}^{T-1}(E^{i}_{t}-\xi^{i}_{t}-\vartheta^{i}_{t}) - \gamma^{i}-\vartheta^{i}_{T}\right]^{+}.
\end{equation*}

To specify risk preferences, we describe agents' risk attitudes by individual utility functions $U^{i}$, $i\in I$, that are assumed to be continuous strictly increasing and concave.
For a random variable $X$, consider the utility functional
 $$
 X \mapsto u^{i}(X)=\E(U^{i}(X)),
 $$
which is defined whenever the expectation exists and is not $-\infty$.
Given an allowance price process $A=\{A_{t}\}_{t=0}^{T}$, agent $i$ behaves rationally
in the sense that it maximizes its utility from the terminal wealth
$$(\vartheta^{i}, \xi^{i}) \mapsto u^{i}(L^{A,i}(\vartheta^{i}, \xi^{i}))$$  by  an appropriate choice
of the strategy that we denote by $(\vartheta^{i*}, \xi^{i*})$.
Following the standard apprehension,  a realistic market state is described by the so-called equilibrium ---
a situation where  the allowance price, positions and abatement measures
are such that each agent is satisfied with the own strategy and,
at the same time, the natural  restrictions are met.
\begin{definition}
\label{de:equilibrium}
An adapted process $A^{*}=\{A^{*}_{t}\}_{t=0}^{T}$
is called an {\rm{equilibrium  allowance price process}} if, for each $i \in  I$,
there is a strategy $(\vartheta^{*i}, \xi^{*i}) \in {\cal U}^{i} $
such that we have
$u^{i}(L^{A^{*}, i} (\vartheta^{*i}, \xi^{*i}))<\infty$  and
\begin{itemize}
\item[{\em(i)}] the cumulative changes in  positions are  in zero  net supply:
\begin{equation*} \label{eqold}
\sum_{i \in I}\vartheta_{t}^{*i}=0 \  \hbox{for all $t=0, \dots, T$,}
\end{equation*}
\item[{\em(ii)}] each agent $i \in  I $ is satisfied with its own strategy in
the sense that, for each $(\vartheta^{i},  \xi^{i}) \in  {\cal U}^{i}$
such that $u^{i}\left(L^{A^{*}, i}
(\vartheta^{i},  \xi^{i})
\right)$ exists, one has
\begin{equation*}
u^{i}\left(L^{A^{*}, i} (\vartheta^{*i}, \xi^{*i}) \right)\ge u^{i}\left(L^{A^{*}, i}
(\vartheta^{i},  \xi^{i})
\right).
\end{equation*}
\end{itemize}
\end{definition}
 In \cite{HinzNovikov}, this equilibrium notion was used to establish
  a  reduced-form model which describes the allowance price evolution
from the risk-neutral perspective. This approach utilizes the following three properties
of the above equilibrium.
    \begin{itemize} \item[(a)]
  There is no  arbitrage since any profitable strategy would immediately be  followed by all agents.
  \item[(b)]
   Given a technology
    with  lower  reduction costs than the present allowance price, it is optimal to immediately  reduce
     one's pollution and take profit from selling allowances.
    \item[(c)]
   There are only two final outcomes for allowance price:
 at maturity, either the price will vanish if there is an excess in allowances, or, in the
 case of their shortage, the price will rise to the penalty level.
 Exact coincidences  of allowance demand and supply at maturity
  occur with zero probability under broad assumptions and can be neglected.
   \end{itemize}

The last property follows from the assumption that the random variable $\sum_i E_T^i$ representing the total emissions at maturity has a continuous distribution, given the information up to time $T-1$.
Under mild additional assumptions, \cite{HinzNovikov} claims
that the  above assertions can be deduced from the equilibrium in the following form.
\begin{prop} \label{mainprop}
Suppose that  $\{A^{*}_{t}\}_{t=0}^{T}$ is an equilibrium allowance price process
and $\{\xi^{*i}_{t}\}_{t=0}^{T-1}$, $i \in I$, are corresponding equilibrium abatement strategies.
\begin{itemize}
\item[{\em(a)}] There exists a measure $\Qb$ on $(\Omega, \cf)$ which is equivalent to $\Pro$ and such that  $\{A^{*}_{t}\}_{t=0}^{T}$ is a $\Qb$-martingale.
\item[{\em(b)}]  For each  $i \in I$ one has
\begin{equation}
\xi^{*i}_{t}= r_{t}^{i}(A^{*}_{t}), \quad t=0, \dots, T-1, \enspace  \label{marginalcosts}
\end{equation}
with the reduction functions $r_{t}^{i}$, $t=0, \dots, T-1$, from {\em(\ref{volumefunction})}.
\item[{\em(c)}]   The terminal value of the allowance price is given by
\begin{equation} \label{termialprice}
A^{*}_{T}=\pi\mathbf{1}\left( \sum_{i \in I}\left(\sum_{t=0}^{T-1}(E^{i}_{t}-\xi^{*i}_{t})- \gamma^{i}\right)  \ge  0 \right).
\end{equation}
where $\mathbf{1}(B)$ is the indicator of the event $B$.
\end{itemize}
\end{prop}
In fact, Proposition \ref{mainprop} states the above-mentioned feedback relationship. Namely,
the equilibrium allowance price process $\{A^{*}_{t}\}_{t=0}^{T}$
is a martingale under  $\Qb\sim \Pro$ that has the terminal value (\ref{termialprice}).
However, this terminal random variable depends  on  the intermediate values  $\{A^{*}_{t}\}_{t=0}^{T-1}$ through
(\ref{marginalcosts}).
A surprising consequence of the feedback relationship is that, from the risk-neutral perspective, only the cumulative
market quantities are relevant. To see this, introduce the overall ``business-as-usual" allowance shortage by
\begin{eqnarray*} \label{capadjusted}
\ce_{T}=\sum_{i \in I}\left(\sum_{t=0}^{T-1}E^{i}_{t} -\gamma^{i}\right).
\end{eqnarray*}
Further, recall from (\ref{cfunt}) the cumulative abatement function $r_t(a)$
to express the risk-neutral certificate price dynamics using (\ref{marginalcosts}) (\ref{termialprice}) and the martingale property of $\{A_t^\ast\}_{t=0}^{T}$ under $\Qb$ as
  $$
A_{t}^\ast=\pi \E^{\Qb}
\left[ \mathbf{1}\left(\ce_{T}-\sum_{t=0}^{T-1}r_{t}(A^{*}_{t}) \ge 0\right)  \, \Big{|} \, \cf_{t}
\right], \qquad t=0, \dots, T-1.
  $$
Although the individual market attributes seem to be irrelevant here,
the  reader should notice that this picture appears only from the
risk-neutral viewpoint.

With this,  the  problem of risk neutral allowance price  modeling boils down to the following task:
\begin{equation*} \boxed{ \begin{array}{c}
\hbox{Given a measure $\Qb \sim \Pro$ and reduction functions } \\
\hbox{$\{r_{t}\}_{t=0}^{T-1}$, describe the random variable $\ce_{T}$ and } \\
\hbox{determine  a $\Qb$-martingale $\{A^{*}_{t}\}_{t=0}^{T}$ such that} \\
 \hbox{$A^{*}_{T}=\pi \mathbf{1}\left(\ce_{T}-\sum_{t=0}^{T-1}r_{t}(A^{*}_{t}) \ge 0\right)$.}
\end{array}} \qquad  \label{ReducedProblem}
\end{equation*}
In this form, transition to the continuous time case is straightforward and the resulting problem can be stated as follows:
\begin{equation}
\boxed{ \begin{array}{c}
\hbox{Given a measure $\Qb \sim \Pro$ and reduction functions } \\
\hbox{$\{r_{t}\}_{t \in [0, T]}$, describe the random variable $\ce_{T}$ and } \\
\hbox{determine  a $\Qb$-martingale $\{A^{*}_{t}\}_{t \in [0, T]}$ such that} \\
 \hbox{$A^{*}_{T}=\pi \mathbf{1}\left(\ce_{T}-\int_{0}^{T}r_{s}(A^{*}_{s})ds \ge 0\right)$.}
\end{array}}
 \label{ReducedProblemCont}
\end{equation}
Problem (\ref{ReducedProblemCont}) is the starting point of our investigation in this paper.
This approach utilizes the ingredients $\{r_{t}\}_{t \in [0, T]}$ and  $\ce_{T}$, which is reasonable from the practical
perspective since the price-dependent reduction functions   $\{r_{t}\}_{t \in [0, T]}$  can be estimated from the market data and the potential  allowance shortage ${\cal E}_{T}$ can be modeled in terms of  emission fluctuations.

\section{Continuous time case: diffusion models}
\label{ContinuousTime}

Modeling in continuous time stipulates that the compliance period is an interval $[0, T]$ and that all the relevant
random evolutions are described  by adapted stochastic processes on a filtered probability space $(\Omega, \cf, \Qb, \{\cf_{t}\}_{t \in [0, T]})$ which is equipped with a ``spot martingale" probability measure $\Qb \sim \Pro$.
Given a random variable  $\ce_{T}$ and appropriate non-decreasing continuous abatement functions $r_{t}: \rr_{+} \times \Omega \to \rr_{+}$ indexed by $t \in [0, T]$, we want to find a solution  $\{A_{t}\}_{t \in [0, T]}$ to
\begin{equation} \label{charact1}
A_{t}=\pi \E^{\Qb}\left[\mathbf{1}\left(\ce_{T}-\int_{0}^{T}r_{s}(A_{s})ds \ge 0\right) \, \Big{|} \, \cf_{t}\right], \qquad t \in [0, T].
\end{equation}

The results of the discrete-time analysis given in \cite{HinzNovikov} suggest that, if the increments of the martingale
\begin{equation*}
\{{\cal E}_{t}:=\E^{\Qb}({\cal E}_{T}\, | \, \cf_{t})\}_{t \in [0, T]}
\end{equation*}
are independent and the abatement functions $r_{t}: \rr_{+} \times\Omega \to \rr_{+}$
are deterministic and time independent, then one can reasonably expect that a solution to (\ref{charact1}) can have the functional form
$$
A_{t}=\alpha(t, X_{t}), \qquad t \in [0, T],
$$
with an appropriate deterministic function
\begin{equation}\label{alpha}
\alpha: \, [0, T]\times \rr \mapsto \rr %\qquad (t, g) \mapsto \alpha(t,g)
 \end{equation}
and a  state  process $\{X_{t}\}_{t \in [0, T]}$ given  by
\begin{equation}\label{processX}
X_{t}:=\ce_{t}-\int_{0}^{t}r_{s}(A_{s})ds, \qquad t \in [0, T].
\end{equation}

In this section, we demonstrate how this approach enables one to find a solution in the framework of diffusion
processes. Assume that $\{W_t\}_{t \in [0, T]}$ is a standard Brownian motion process (under $\Qb \sim \Pro$) and
that our $\{\cf_{t}\}_{t \in [0, T]}$ is the natural filtration of the process. In this case, by the martingale
representation theorem \cite{KaratzsasShreve1}, one must have
$$
d\ce_{t}=\sigma_{t}dW_{t}
$$
for some admissible adapted process $\{\sigma_t\}_{t \in [0, T]}$. To ensure that $\{\ce_{t}\}_{t \in [0, T]}$ has independent increments, we assume that $\{\sigma_{t}=\sigma(t)\}_{t \in [0, T]}$ is a known deterministic function and that we are given
continuous non-decreasing time-independent abatement functions $\{r_{t}=r\}_{t \in [0, T]}$.
To verify the martingale property of the allowance price process
$$
A_{t}=\alpha(t, X_{t}), \qquad t \in [0, T],
$$
we use It\^o's formula and (\ref{processX}) to write the stochastic differential of the process as
\begin{align*}
dA_{t} &= d\alpha(t, X_{t})\\
 &= \partial_{(1, 0)}\alpha(t, X_{t})dt+ \partial_{(0, 1)}\alpha(t, X_{t})d X_{t}
            + \frac{1}{2}\partial_{(0, 2)}\alpha(t, X_{t})d[\ce]_{t}\\
 &=  \partial_{(1, 0)}\alpha(t, X_{t})dt-
\partial_{(0, 1)}\alpha(t, X_{t})r(\alpha(t,X_{t}))dt+\frac{1}{2}\partial_{(0, 2)}\alpha(t, X_{t})\sigma^{2}(t)dt \\
& \hspace{.5cm}+ \partial_{(0, 1)}\alpha(t, X_{t}) \sigma(t)dW_{t}.
\end{align*}
Here $[\ce]_{t}$ stands for the quadratic variation of the martingale
$\{\ce_{t}\}_{t \in [0, T]}$ and $\partial_{(i, j)}$ denotes the respective partial derivatives.
Now we observe that
the function $\alpha$ is a  solution to
\begin{equation} \label{pde} \
\partial_{(1, 0)}\alpha(t, x)-
r(\alpha(t,x))\partial_{(0, 1)}\alpha(t, x)+\frac{1}{2}\sigma^{2}(t)\partial_{(0, 2)}\alpha(t, x)=0
\end{equation}
in $(0,T)\times\R$ with the boundary condition
\begin{equation} \label{bc}
\hbox{$\alpha(T, x)=\pi \mathbf{1}(x\geqslant 0), \qquad x \in \rr$,}
\end{equation}
justified by the digital terminal allowance price, then the thus constructed $\{A_t\}_{t \in [0, T]}$ will indeed be a martingale that satisfies (\ref{charact1}) by construction. Note that $\alpha(t,x)$ satisfies the maximum principle, so that $0\leqslant \alpha(t,x) \leqslant\pi$ for all $(t,x)\in[0,T]\times\R$. For the proof of this fact, see Proposition \ref{maxprinciple} below.

The following summarizes the above-presented approach.

\bigskip

{\bf Allowance price in the diffusion framework}%
\begin{enumerate}
\item Given a continuous non-decreasing function $r : [0, \infty) \to [0, \infty)$
      and a positive function $\sigma(t)$, $t \in [0, T]\times\R$, determine a solution $\alpha$ to the boundary value problem (\ref{pde}), (\ref{bc}). (We assume that $\sigma(t)$ is regular enough to ensure existence and uniqueness of the solution.)
\item Verify that there is a unique strong solution to
\begin{equation}
dX_{t}=d\ce_{t}-r(\alpha(t, X_{t}))dt, \qquad X_{0}=\ce_{0}. \label{sde}
\end{equation}
\item Introduce the allowance price $\{A_{t}\}_{t \in [0, T]}$ by
\begin{equation*}
A_{t}:=\alpha(t, X_{t}), \qquad t \in [0, T]. \label{introd}
\end{equation*}
\end{enumerate}

Having  constructed the allowance price process $\{A_{t}\}_{t \in [0, T]}$  in this way, one obtains a standard procedure for the valuation of European options.
%Moreover, these contracts can be priced directly, without solving the stochastic differential equation (\ref{sde}).
Indeed, observe that, due to the Markov property of the strong solution to $(\ref{sde})$,
the fair time $t$ price of a European call option written on the allowance price at (maturity) date $\tau \in (t, T]$ is given in terms of an appropriate function of the state variable:
$$
C_{t}=\E^{\Qb}\left[(A_{\tau}-K)^{+} \mid  \cf_{t}\right]=\E^{\Qb}\left[(\alpha(\tau, X_\tau)-K)^{+} \mid \cf_{t}\right]=f^{\tau}(t, X_t).
$$
To ensure that $\{C_{t}=f^{\tau}(t, X_t)\}_{t \in [0, \tau]}$ is a martingale,
the function $f^{\tau}: [0, \tau) \times \rr \to \rr$ is to be taken as a solution to the linear partial differential equation
\begin{equation} \label{pde1}
\partial_{(1, 0)}f^{\tau}(t, x)-
\partial_{(0, 1)}f^{\tau}(t, x)r(\alpha(t,x))+\frac{1}{2}\partial_{(0, 2)}f^{\tau}(t,x)\sigma^{2}(t)=0,
\end{equation}
in  $(0,\tau)\times\R.$
However, the boundary condition in this case will be
\begin{equation} \label{bc1}
\hbox{$f^{\tau}(\tau, x)=(\alpha(\tau,x)-K)^{+}$, \qquad $x \in \rr$.}
\end{equation}

Summarizing, we obtain the following description for the procedure.

\bigskip

{\bf Valuating a European call in the diffusion framework}
\begin{enumerate}
\item Find the function $\alpha$ as above.
\item Given a strike price $K\ge 0$ and maturity time $\tau \in [0, T]$ of a European call,
      calculate $f^{\tau}$ as the solution to the boundary problem (\ref{pde1}), (\ref{bc1}).
\item Given a time $t \in [0, \tau]$ and the allowance price $a \in [0, \pi]$ at time $t$,
      obtain  $x$ as the solution to $\alpha(t,x)=a$.
\item Substitute $t$ and the thus obtained $x$ into the function $f^{\tau}$ to obtain the time $t$ price of the European call as $f^{\tau}(t, x)$.
\end{enumerate}

Note that one can also estimate $f^{\tau}$ directly using Monte Carlo simulations: given a strike price $K\ge 0$, a maturity time $\tau \in [0, T]$ of a European call, a time $t \in [0, \tau]$ and the allowance price $a \in [0, \pi]$ at time $t$, obtain $X_t= x$ as the solution to $\alpha(t,x)=a$. Then, using the Markov property of $\{X_t\}$, one can evaluate the option price
$$
C_{t}=\E^{\Qb}\left[(\alpha(\tau, X_\tau)-K)^{+} \mid X_t=x\right]
$$
by estimating the expectation via simulating a (large enough) number $N_{MC}$ of copies of the random variable $X_\tau$.

Note that closed-form solutions to the non-linear  partial differential equation (\ref{pde})
are rarely available. However, a linear abatement function leads to explicit expressions, as
pointed out in \cite{Seifert} and \cite{Wagner}. We will consider this case as an illustration. \\[10pt]
{\bf Example}.
Given a linear abatement function $r(a)= ca$, $c \in (0, \infty)$,
and a constant diffusion coefficient $\{\sigma_{t}=\sigma\}_{t \in [0, T]}$ with $\sigma \in (0, \infty)$,
the partial differential equation (\ref{pde}) becomes Burger's equation
\begin{equation}
\partial_{(1, 0)}u- c u \partial_{(0,1)}u+\frac{\sigma^{2}}{2} \partial_{(0,2)}u =0,  \label{Burger}
\end{equation}
whose solution  can be obtained from that of the heat equation via the Hopf-Cole transform.

Namely, a direct calculation shows that if $v: [0, T]\times \rr \mapsto \rr$ solves the heat equation
\begin{equation}
\partial_{(1,0)}v + \frac{\sigma^{2}}{2} \partial_{(0,2)}v=0, \label{Heat}
\end{equation}
then  its Hopf-Cole transform
\begin{equation}
u:=- \frac{\sigma^{2}}{c} \frac{\partial_{(0, 1)}v}{v} \label{HopfCole}
\end{equation}
solves (\ref{Burger}). In order to satisfy an original boundary condition
$$
u(T, x)\equiv-\frac{\sigma^{2}}{c} \frac{\partial_{(0, 1)}v(T, x)}{v(T, x)}=b(x), \qquad x \in \rr,
$$
the boundary value function for (\ref{Heat}) must be chosen as
\begin{equation*}
v(T, x)=\tilde b(x):= \exp\left\{ -\frac{c}{\sigma^{2}} \int_{-\infty}^{x}b(u)du \right\}, \qquad x \in \rr. \label{chosen}
\end{equation*}
Our digital  boundary condition has the form
$$
u(T, x)=b(x)=\pi \mathbf{1}(x\geqslant 0), \qquad x\in\rr,
$$
so that we should take
$$
v(T, x)=\tilde b(x)=\mathbf{1}(x< 0)+  \mathbf{1}(x\geqslant 0)e^{-c \pi x/\sigma^{2}}, \qquad x\in\rr.
$$
Denoting by $\NN(\mu,s^2)$ the normal distribution with mean $\mu$ and variance $s^2$ and by $\Phi(y)$ and $\varphi(y)=\Phi'(y)$ the standard normal distribution function and its density, respectively, we can write:
\begin{align}
v(t,x)&=\int\tilde b(y)\NN(x, \sigma^{2}(T-t))(dy) \nonumber\\%\label{toplug} \\
&=\int_{-\infty}^{0}\NN(x,\sigma^{2}(T-t))(dy)+\int_{0}^{\infty}e^{-c\pi y/\sigma^{2}}\NN(x, \sigma^{2}(T-t))(dy)
\nonumber \\
&=\Phi\left(\frac{-x}{\sigma\sqrt{T-t}}\right)+\Phi\left(\frac{x-c\pi(T-t)}{\sigma \sqrt{T-t}}\right)
\exp\left\{-\frac{c \pi}{\sigma^{2}}x+ \frac{\pi^{2}c^{2}(T-t)}{2\sigma^{2}}\right\}. \label{orig}
\end{align}
The derivative of this expression with respect to $x$ is equal to
\begin{align}
\partial_{(0,1)}v(t,x)&=-\frac{\varphi\left(\frac{-x}{\sigma\sqrt{T-t}}\right)}{\sigma \sqrt{T-t}}+
\frac{\varphi\left(\frac{x-c\pi(T-t)}{\sigma \sqrt{T-t}}\right)}{\sigma\sqrt{T-t}}
\exp\left\{-\frac{c \pi}{\sigma^{2}}x+ \frac{\pi^{2}c^{2}(T-t)}{2\sigma^{2}}\right\} \nonumber \\
& -\frac{c \pi}{\sigma^{2}}\Phi\left(\frac{x-c\pi(T-t)}{\sigma \sqrt{T-t}}\right)
\exp\left\{-\frac{c \pi}{\sigma^{2}}x+ \frac{\pi^{2}c^{2}(T-t)}{2\sigma^{2}}\right\}. \label{deriv}
\end{align}
Next we use (\ref{HopfCole}) to calculate
\begin{equation} \label{shape}
\alpha(t,x)=-\frac{\sigma^{2}}{c} \frac{\partial_{(0,1)}v(t,x)}{v(t,x)}, \qquad (t,x)\in [0, T]\times\R.
\end{equation}
The  shape of this function is depicted in Figure \ref{pic}.
\begin{figure}
\begin{center}
   \includegraphics[height=2.2in, width=5in]{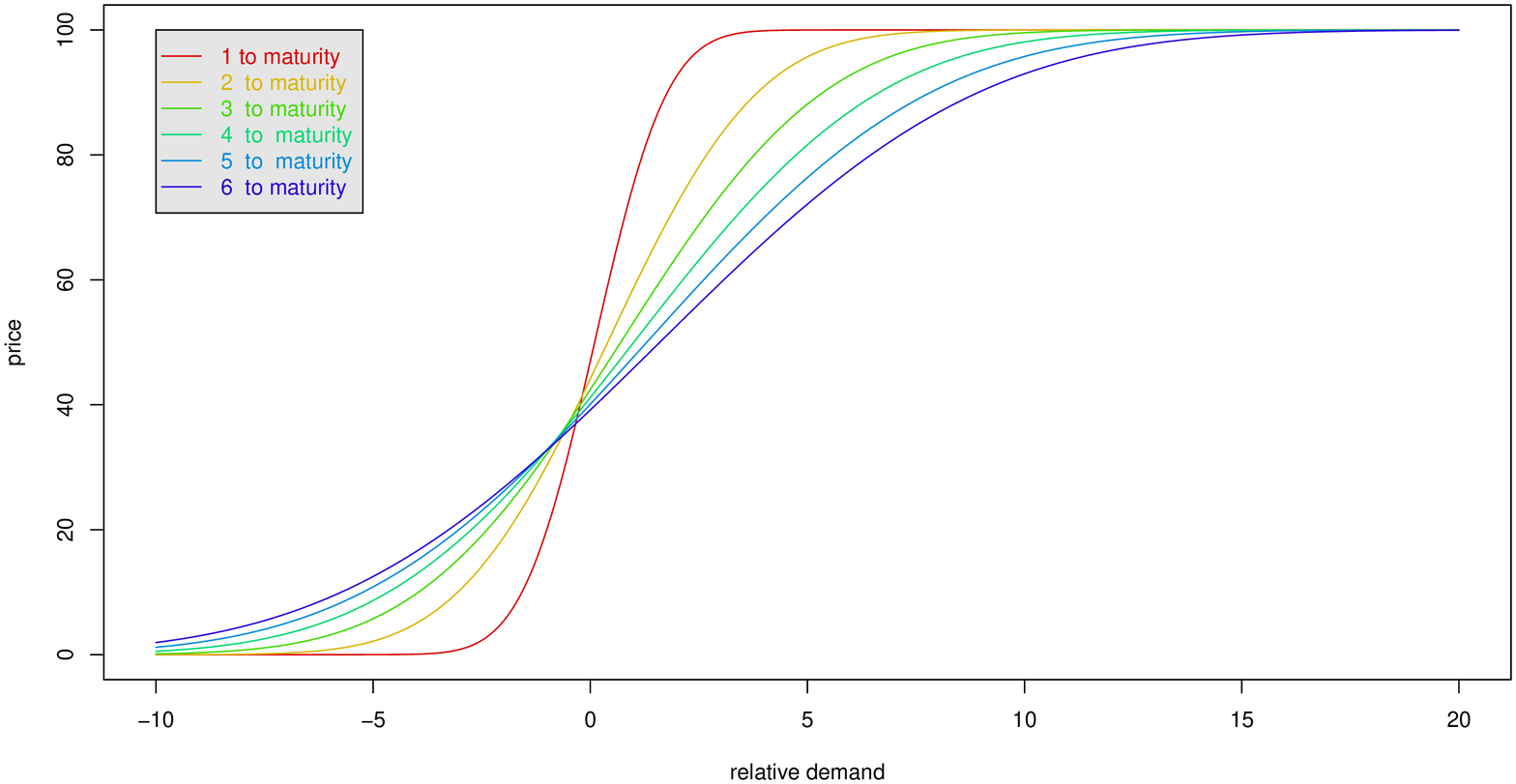}
   \caption{ \label{pic}
   The  functions $\alpha(t, \cdot )$ displayed for $t=1.9, 1.6,1.3,1.0,0.7,0.4$ and calculated using
   (\ref{orig})--(\ref{shape})
   with $T=2$, $\sigma=4$, $c=0.02$ and $\pi=100$.}
\end{center}
\end{figure}
%\\[10pt] To value a European call, Burger's equation
%must be solved with a new boundary condition.
%Thus, we conclude that, given the time  $t \in  [0, \tau]$ allowance price
%$a=\alpha(t, x)$, the fair value
%of a European call with  maturity $\tau \in [t, T]$ and strike price $K$ is given by
%$$
%f^{\tau}(t,x)=-
%\frac{\displaystyle\sigma^{2}\int\tilde b(y)\NN(x, \sigma^{2}(\tau-t))(dy)}
%{\displaystyle c\int\tilde b(y)\frac{y-x}{\sigma^{2}(\tau-t)}\NN(x, \sigma^{2}(\tau-t))(dy)}
%$$
%with
%$$
%\tilde b(x)= \exp\left\{-\frac{c}{\sigma^{2}} \int_{-\infty}^{x}\left[
%-\frac{\sigma^{2} \partial_{(0,1)}v(\tau, u)}{c v(\tau, u)}
%-K\right]^{+}du\right\},
%$$
%where $v(\tau, x)$ and $\partial_{(0,1)}v(\tau, x)$ result from  (\ref{orig}), (\ref{deriv})
%(replacing $t$ by $\tau$):
%\begin{eqnarray*}
%v(\tau,x)
%=\Phi\left(\frac{-x}{\sigma\sqrt{T-\tau}}\right)+\Phi\left(\frac{x-c\pi(T-\tau)}{\sigma \sqrt{T-\tau}}\right)
%\exp\left\{-\frac{c \pi}{\sigma^{2}}x+ \frac{\pi^{2}c^{2}(T-\tau)}{2\sigma^{2}}\right\}, \\
%\partial_{(0,1)}v(\tau,x)=\frac{-\varphi\left(\frac{-x}{\sigma\sqrt{T-\tau}}\right)}{\sigma \sqrt{T-\tau}}+
%\frac{\varphi\left(\frac{x-c\pi(T-\tau)}{\sigma \sqrt{T-\tau}}\right)}{\sigma\sqrt{T-\tau}}
%\exp\left\{-\frac{c \pi}{\sigma^{2}}x+ \frac{\pi^{2}c^{2}(T-\tau)}{2\sigma^{2}}\right\} \nonumber \\
% -\frac{c \pi}{\sigma^{2}}\Phi\left(\frac{x-c\pi(T-\tau)}{\sigma \sqrt{T-\tau}}\right)
%\exp\left\{-\frac{c \pi}{\sigma^{2}}x+ \frac{\pi^{2}c^{2}(T-\tau)}{2\sigma^{2}}\right\}.
%\end{eqnarray*}

To illustrate our valuation procedure, consider the following parameters:
time to compliance date $T=2$, diffusion coefficient $\sigma=4$, penalty $\pi=100$ and a linear abatement function $r(a)=ca$ with $c=0.02$. Taking $t=0$, we consider a family
of European calls with the same strike price $K=25$, but different maturity times $\tau $.
Suppose that the initial allowance price is equal to the strike price $a=A_{0}=25$ (the so-called
at-the-money situation), which is attained by $\alpha(0,x)=25$ with $x \approx -2.434$.
Next, we determine the call prices $C(0, \tau)$ at time $t=0$
for different maturity times $\tau \in [0, T]$. Independently of the model, the
price of an expiring call with $\tau=0$ must be equal to zero, whereas the longest-maturity
call with $\tau=T$ must have the price  $A_{0}(\pi-K)/\pi=25\times 0.75=18.75$. (Note that, because of
the  digital payoff, such a call  is equivalent to  $0.75$ allowances). That is, the call prices increase
with contract's maturity from $0$ to $18.75$.
The shape of this curve is obtained using a crude Monte Carlo procedure. We have simulated $N_{MC}=10^4$ i.i.d. copies of the random variable $X_\tau$. We used the forward Euler method to generate i.i.d. copies of $\{X_t\}$, with time step $0.02$. The curve is presented in Figure \ref{callline}, together with $95\%$ confidence intervals.

\begin{figure}
\begin{center}
   \includegraphics[height=2.2in, width=5in]{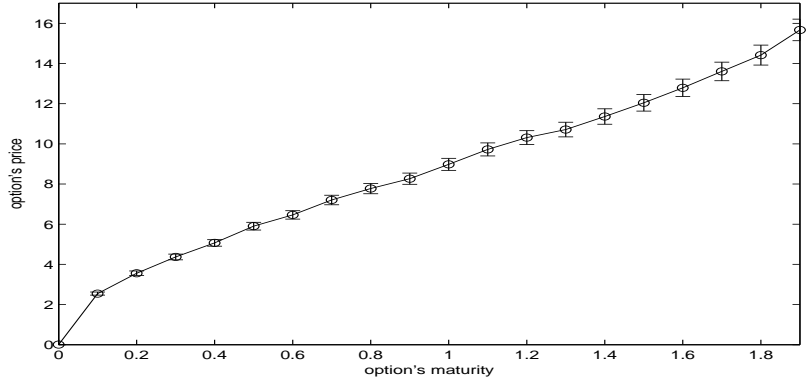}
   \caption{ \label{callline}
   The prices of European calls depending on their maturity times.}
\end{center}
\end{figure}

\section{Allowance price in the jump-diffusion setting}
\label{Jumps}

In this section, we describe a situation where the market can suddenly change due to the presence of jumps. We propose a framework suitable to the present practice of the EU ETS where the member states negotiate  their allowance allocations and the market needs to adapt to the new situation. For instance, revised decision on the amount of allocated certificates yields a jump of allowance market prices. Also, a sudden change in demand and/or price for fuel can result in the pollution levels changing dramatically which impacts on the allowance prices.

We discuss how to solve problem (\ref{ReducedProblemCont}) in the jump-diffusion setting using
a similar approach to the one employed in Section \ref{ContinuousTime}. Namely, we choose a candidate
for (\ref{charact1}) of the form
\begin{equation*}
A_t= \alpha(t,X_t),
\end{equation*}
where $\alpha$ and $X_t$ are as in (\ref{alpha}) and (\ref{processX}), respectively. But now we assume that the martingale
$${\cal E}_{t}:=\E^{\Qb}\left[{\cal E}_{T}\, | \, \cf_{t} \right], \qquad t \in [0, T]$$ is modeled
using a general jump-diffusion process adapted to a filtration $\{{\cal F}_t\}_{t \in [0, T]}$. In this setting, the process $\{X_t\}_{t\in[0,T]}$ is given by the following stochastic differential equation:
\begin{equation}\label{diffGt}
dX_t = -r(\alpha(t,X_t))dt + \sigma(t,X_t)dW_t + \int_{\R_0} a(t-,X_{t-},y)(p_\nu-q_\nu)(dy\times dt),
\end{equation}
where  $\R_0=\R\setminus \{0\}$, $\{W_t\}_{t \in [0, T]}$ is a Brownian motion adapted to the filtration $\{{\cal F}_t\}_{t \in [0, T]}$, and $p_\nu$ is an independent of $\{W_t\}_{t \in [0, T]}$ $\{{\cal F}_t\}_{t \in [0, T]}$-adapted random Poisson measure with intensity $q_\nu(dy\times dt)=\lambda\nu(dy)dt$, $\nu$ being a probability distribution on $\R$ and $\lambda \in (0, \infty)$
a positive constant. Expression (\ref{diffGt}) allows us to model jumps
with great flexibility, as the distribution of the jump is state and time dependent. It is known that, under suitable
Lipschitz and growth conditions on $r$, $\sigma$ and $a$, stochastic differential equation (\ref{diffGt}) possesses
a unique strong solution \cite{GikhmanS}.

We again use It\^o's lemma for $A_t=\alpha(t,X_t)$ to find conditions on $\alpha$ under which the process $\{A_t\}_{t \in [0, T]}$ will be a martingale under $\Qb$. Let $\tau=T-t$. For convenience, we introduce a new function $\beta(\tau,x)=\alpha(T-\tau,x)$ and, instead of (\ref{pde}), obtain the following nonlinear partial integro-differential equation for $\beta$:

\begin{eqnarray}
\lefteqn{\partial_{(1,0)}\beta(\tau,x)= -r(\beta(\tau,x))\partial_{(0,1)}\beta(\tau,x)+\frac{1}{2}\sigma^2(\tau,x)\partial_{(0,2)}\beta(\tau,x)}  \nonumber\\
\label{pide}
&&  + \lambda\int \left[\beta(\tau,x+a(\tau,x,y))-\beta(\tau,x)-a(\tau,x,y)\partial_{(0,1)}\beta(\tau,x)\right]\nu(dy),
\end{eqnarray}
$(\tau,x)\in[0,T]\times\R$, with the boundary condition
\begin{equation}\label{pidebc}
\beta(0,x)= \pi \mathbf{1}(x\geqslant 0),  \qquad x\in\R.
\end{equation}

Under certain assumptions, jump-diffusion models for option pricing lead to partial integro-differential equations, see e.g.\ \cite{Andersen, Cont}. Equation (\ref{pide}) differs from them in that it has a nonlinear coefficient $r(\beta(\tau,x))$. We assume that there exists a classical solution $\beta\in C^{1,2}([0,T)\times \R)$ to the problem (\ref{pide}), (\ref{pidebc}). It is not difficult to see that it satisfies the maximum principle.

\begin{Proposition}\label{maxprinciple}
Solution $\beta(\tau,x)$ to equation (\ref{pide}) with initial boundary condition $\beta(0,x)=h(x)$ satisfies the maximum principle:
\begin{equation}\label{mp}
\inf_z h(z) \leqslant \beta(\tau,x) \leqslant \sup_z h(z), \qquad  (\tau,x)\in [0,T]\times\R.
\end{equation}
\end{Proposition}

\begin{proof}
Since with $\alpha(t,x)=\beta(T-t,x)$, where $\beta$ satisfies (\ref{pide}), the process $A_t=\alpha(t,X_t)$ is a $\Qb$-martingale. Also, the process $\{X_t\}$ given by (\ref{diffGt}) is a Markov process. Thus, one has a.s.
\[
\alpha(t,X_t)=\E(A_T\mid {\cal F}_t)= \E(\alpha(T,X_T)\mid {\cal F}_t)=\E(\alpha(T,X_T)\mid X_t)
\]
and so, due to regularity of $\alpha$, for $\tau=T-t$,
\begin{align*}
\beta(\tau,x)&=\alpha(t,x)=\E(\alpha(T,X_T)\mid X_t=x)\\
&=\E(\beta(0,X_T)\mid X_t=x) =\E(h(X_T)\mid X_t=x),
\end{align*}
which immediately implies (\ref{mp}).
\end{proof}

Note that Proposition \ref{maxprinciple} also proves the maximum principle for solution to the boundary value problem (\ref{pde}), (\ref{bc}).

We will see that, under a mild condition, the maximum principle remains valid for the discretized problem as well (see Proposition \ref{maxp} below).

Analytical solution to partial integro-differential equations can only be obtained in a few special cases, so in most situations one can only solve them numerically.
In this section, we discuss a discretization of (\ref{pide}) using the finite difference method and prove that the discretized equation has a unique solution. Finally, we illustrate the approach by providing a numerical example.

First we need to truncate the domain of $x$. Let $D_l:=\{x\in\R:\, |x|<l\}$, where the bound $l>0$ of the domain can be chosen so that the probability that the process $\{X_t\}$ leaves $D_l$ during the time interval $[0,T]$ given it starts at $X_0$ does not exceed a given small $\epsilon>0$. This procedure will be illustrated at the end of the section.

Next, we restrict the domain of integration for the integral term on the RHS of (\ref{pide}) to an interval $[K_1,K_2]$, chosen such that the error made due to the truncation also remains small. For a good choice of the terminals $K_j$, one can refer to the study presented in \cite{Bria04}.

The numerical solution will be computed on a discrete grid. Let $N$ denote the total number of discrete $x$-values and $M$ the total number of $\tau$-values we want to use in the grid for the numerical solution, so that the step sizes in $x$ and $\tau$ are respectively $\Delta_{x}=2l/N$ and $\Delta_\tau=T/M$. Put $x_i=-l+i\Delta_{x}$, $\tau_n=n\Delta_{\tau}$, for $i\in\Z$ and $n=0,\hdots,M$. We use $\beta_i^n=\beta(\tau_n, x_i)$ for the values of $\beta$ on this grid. Because of the presence of the non local term on the RHS of equation (\ref{pide}), one needs to define $\beta$ outside $[0,T]\times D_l$. We chose the simplest and most intuitive approach and set $\beta_i^n:=g(x_i)$ for $i\notin\{0,\hdots,N-1\}$, where
\begin{equation*}
g(x):=\begin{cases}
\pi ~~\textrm{ if } ~~x\geqslant l,\\
0 ~~\textrm{ if } ~~x\leqslant -l.
\end{cases}
\end{equation*}
The partial derivatives are replaced by the respective finite differences:
\begin{align*}
 \partial_{(1,0)}\beta
(\tau_{n}, x_{i}) &\approx \frac{\beta_i^{n+1}-\beta_i^n}{\Delta_{\tau}},\\
\partial_{(0,1)}\beta
(\tau_{n}, x_{i}) &\approx \frac{\beta_i^n-\beta_{i-1}^n}{\Delta_{x}},\\
\partial_{(1,0)}\beta
(\tau_{n}, x_{i})
&\approx \frac{\beta_{i+1}^n-2\beta_i^n + \beta_{i-1}^n}{(\Delta_{x})^2}.
\end{align*}
Following an approach similar to the one used in \cite{Cont05}, we consider the same step size $\Delta_{x}$ to approximate the integral term, and choose $J_1$ and $J_2$ such that $[K_1,K_2]\subseteq [(J_1-1/2)\Delta_{x}, (J_2+1/2)\Delta_{x}]$, which leads to the approximation
\[
\lambda\int_{\R_0} \beta(\tau_n,x_i+a(\tau_n,x_i,y)) \nu(dy) \approx \lambda\sum\limits_{j=J_1}^{J_2} \nu_j \beta_{j^\ast(i,j,n)}^n
\]
where
\[
\nu_j := \int\limits_{(j-1/2)\Delta_{x}}^{(j+1/2)\Delta_{x}} \nu(dy)
\]
and
\[
j^\ast(i,j,n) := \mbox{ arg } \min_k |(x_i+ a(\tau_n,x_i,x_j))-(k \Delta_{x}-l)|.
\]
Similarly, one has
\[
\partial_{(0,1)}\beta(\tau_n,x_i)\int\limits_{K_1}^{K_2}a(\tau_n,x_i,y)\nu(dy) \approx \frac{\beta_i^n-\beta_{i-1}^n}{\Delta_{x}}\sum\limits_{j=J_1}^{J_2} a_{i,j}^n\nu_j,
\]
where $a_{i,j}^n := a(\tau_n,x_i,x_j)$. We set $\Sigma_i^n := \sum\limits_{j=J_1}^{J_2} a_{i,j}^n\nu_j$ for convenience.

So now we are looking for a solution $\{\beta_i^n\}$ on the grid such that, for $n=0,\hdots,M-1$,
\begin{align} \nonumber
\frac{\beta_i^{n+1}-\beta_i^n}{\Delta_{\tau}}&=
\frac{1}{2}\sigma^2(\tau_{n+1},x_i) \frac{\beta_{i+1}^{n+1}-2\beta_i^{n+1} + \beta_{i-1}^{n+1}}{(\Delta_{x})^2}-r(\beta^{n}_{i})\frac{\beta_i^{n+1}-\beta_{i-1}^{n+1}}{\Delta_{x}} \\
&
+ \lambda \sum\limits_{j=J_l}^{J_r} \nu_j \beta_{j^\ast(i,j,n)}^{n}- \lambda\beta_i^{n+1}
-\lambda \frac{\beta_i^{n+1}-\beta_{i-1}^{n+1}}{\Delta_{x}}\Sigma_i^{n+1} \label{discrete_eq}
\end{align}
with
\begin{equation*}
\beta_i^{0} = \pi \mathbf{1}(x_i \geqslant 0)\quad \hbox{for $ i\in \zz$},  \qquad \beta_i^{n}=g(x_{i}) \quad \hbox{for $i<0$ and $i\ge N$.}
\end{equation*}

Let $\Sigma^\ast := \min \Sigma_i^n$ and $\sigma^{\ast 2}:=\min \sigma^2(\tau_n,x_i)$, where the minima are taken over all $i\in\{0,\hdots,N-1\}$ and $n\in\{0,\hdots,M\}$.

\begin{Proposition}\label{maxp}
The discretized problem has a unique solution $\{\beta_i^n\}$. If, in addition, $\sigma^2$ is bounded away from zero and the discrete grid is such that
\begin{equation*}\label{cond}
-\Sigma^\ast\Delta_{x} \leqslant \frac{\sigma^{\ast 2}}{2\lambda },
\end{equation*}
then the solution satisfies the maximum principle:
\[
0\leqslant \beta_i^n\leqslant\pi \mbox{ for any } i\in\Z \mbox{ and } n\in\{0,\hdots,M-1\},
\]
where $\pi$ is the penalty per unit of pollutant not covered by the initial allocation.
\end{Proposition}

This proposition ensures that the discretized problem is well posed. In particular, it shows that the allowance price process is always positive and does not exceed the level of penalty fixed by the regulator at any time during the compliance period.

\begin{proof}
We follow the steps in the proof from \cite{Cont05}.
For $n\in\{0,\hdots,M-1\}$ and $i\in\{0,\hdots, N-1\}$, equation (\ref{discrete_eq}) can be rewritten as
\begin{eqnarray}\label{newproblem} \lefteqn{
-F(\beta_i^n)\Delta_{\tau} \beta_{i-1}^{n+1} + (1+G_{i}^{n}(\beta_i^n)\Delta_{\tau})\beta_i^{n+1} - H_{i}^{n}\Delta_{\tau}\beta_{i+1}^{n+1} }
\qquad\qquad \qquad  \qquad\\
& & \qquad \qquad \qquad
=  \beta_i^n + \lambda\Delta_{\tau} \sum\limits_{j=J_l}^{J_r} \nu_j \beta^n_{j^\ast(i,j,n)}, \nonumber
\end{eqnarray}
where
\begin{align}
F_{i}^{n}(\beta_i^n) &= \frac{\sigma^2(\tau_{n+1},x_i)}{2(\Delta_{x})^2} + \frac{r(\beta_i^n)}{\Delta_{x}} + \frac{\lambda \Sigma_i^{n+1}}{\Delta_{x}} , \label{coeff1}\\
G_{i}^{n}(\beta_i^n) &= \frac{\sigma^2(\tau_{n+1},x_i)}{(\Delta_{x})^2} + \frac{r(\beta_i^n)}{\Delta_{x}} + \frac{\lambda \Sigma_i^{n+1}}{\Delta_{x}} +\lambda , \label{coeff2}\\
H_{i}^{n}&= \frac{\sigma^2(\tau_{n+1},x_i)}{2(\Delta_{x})^2}. \label{coeff3}
\end{align}This is a linear system and can be written as
\begin{equation}\label{linearsyst}
\mathbf{M}(n)\mathbb{\beta}^{n+1} = \mathbf{y}_n \quad \hbox{for all }\quad n=0,\hdots,M-1,
\end{equation}
where $\mathbf{\beta}^{n+1}=(\beta_0^{n+1},\hdots,\beta_{N-1}^{n+1})\in\R^N$, $\mathbf{y}_n=(y_0^n,\hdots,y_{N-1}^n)\in\R^N$ with components
\[
y^{n}_{i}= \beta_i^n + \lambda\Delta_{\tau} \sum\limits_{j=J_l}^{J_r} \nu_j \beta^n_{j^\ast(i,j,n)},
\qquad i=0, \dots, N-2
\]
and
\[
y^{n}_{N-1}=\beta_{N-1}^n + \lambda\Delta_{\tau} \sum\limits_{j=J_l}^{J_r} \nu_j \beta^n_{j^\ast(N-1,j,n)}
+ H_{N-1}^{n}\Delta_{\tau} \pi
\]
 given by the right hand side of (\ref{newproblem}).
 Given $\beta^{n}$, the matrix $\mathbf{M}(n)\in\R^{N\times N}$ is tridiagonal: the elements on
its main diagonal are the terms $(1+G_{i}^{n}(\beta_i^n)\Delta_{\tau})$, the elements of the fist diagonal above it are given by $-H_{i}^{n}\Delta_{\tau}$ and the elements of the fist diagonal below the main diagonal are given by $-F_{i}^{n}(\beta_i^n)\Delta_{\tau}$. Furthermore,  $\mathbf{M}(n)$ is diagonally dominant, which can be seen  from the relation
$$
\hbox{$G_{i}^{n}(\beta_i^n)=F_{i}^{n}(\beta_i^n)+H_{i}^{n}+\lambda\quad$ for $\quad i=0, \dots, N-1$
and $n=0, \dots M$ }
$$
and the non-negativity of the coefficients (\ref{coeff1})--(\ref{coeff3}).
Therefore, given $\beta^{n}$, the linear system (\ref{linearsyst}) possesses a unique solution $\beta^{n+1}$, see \cite{Lancaster}.
Hence the existence and uniqueness of the solution to  (\ref{discrete_eq}) follows by induction.

Now we will use induction in $n$ to show that the maximum principle holds. We will only prove that $\beta_i^n$ are non-negative, as the argument can be easily adapted to prove that the values $\beta_i^n$ remain bounded by $\pi$.

We want to show that, for any $\Delta_{\tau}>0$ and $\Delta_{x}>0$,
\[
\beta_i^n\geqslant 0 \quad\mbox{ for all }\quad i\in\Z \mbox{ and } n\in\{0,\hdots,M-1\}.
\]
For $n=0$, this is obvious from the shape of the boundary condition (\ref{pidebc}) and the definition of $g$. For induction step, assume that $\beta^n_i\geqslant 0$ for all $i\in\Z$, but there exists an $i_0\in\Z$ such that $\beta_{i_0}^{n+1}<0$. By definition of $g$, $i_0\in\{0,\hdots,N-1\}$ since we would have $\beta_{i_0}^{n+1}=g(x_{i_0})\geqslant 0$ otherwise. One can choose $i_0$ such that
\[
\beta_{i_0}^{n+1}=\min_{i\in\{0,\hdots,N-1\}} \beta_i^{n+1}<0.
\]
Under assumption (\ref{cond}), $F_{i}^{n}$, $G_{i}^{n}$ and $H_{i}^{n}$ are all non-negative.
It follows from $G_{i}^{n}(\beta_i^n)=F_{i}^{n}(\beta_i^n)+H_{i}^{n}+\lambda$ that
\begin{align*}
\beta_{i_0}^{n+1}&= -F_{i}^{n}(\beta_{i_0}^n)\Delta_{\tau} \beta_{i_0}^{n+1} + (1+G_{i}^{n}(\beta_{i_0}^n)\Delta_{\tau})\beta_{i_0}^{n+1} - H_{i}^{n}\Delta_{\tau}\beta_{i_0}^{n+1} - \lambda\Delta_{\tau} \beta_{i_0}^{n+1}\\
&\geqslant -F_{i}^{n}(\beta_{i_0}^n)\Delta_{\tau} \beta_{i_0-1}^{n+1} + (1+G_{i}^{n}(\beta_{i_0}^n)\Delta_{\tau})\beta_{i_0}^{n+1} - H_{i}^{n}\Delta_{\tau}\beta_{i_0+1}^{n+1}\\
&= \beta_{i_0}^n + \lambda\Delta_{\tau} \sum\limits_{j=J_l}^{J_r} \nu_j \beta^n_{j^\ast(i_0,j,n)}\geqslant 0,
\end{align*}
which is a contradiction. The proposition is proved.
\end{proof}

In conclusion of this section, we give a method for choosing the domain boundary $l$ and provide a numerical illustration. Let us focus on a special case where  $a(\tau,x,y)=y$ holds for all $y \in \rr$, which corresponds to a compensated compound Poisson process.
Under this assumption, $\{X_t\}_{t \in [0, T]}$ follows a jump-diffusion process given by
\begin{equation*}\label{jumpdiffusion}
X_t=X_{0}  -\int_0^t r(\alpha(s,X_s))ds + \int_0^t \sigma(s,X_s)dW_s + \sum\limits_{j=1}^{N_t} Y_i - \lambda \E(Y_1)t,
\end{equation*}
where $\{N_t\}_{t \in [0, \infty[}$ stands for a Poisson process with intensity $\lambda$ and $\{Y_i\}_{i \in \nn}$ is a sequence of independent identically distributed random variables following the same $\nu = N(0, 1)$, all the components $\{W_t\}$, $\{N_t\}$ and $\{Y_i\}$ of the model being independent of each other. Denote the compensated jump component by
\[
J_t := \sum_{j=1}^{N_t} Y_i - \lambda \E(Y_1)t = \sum_{j=1}^{N_t} Y_i
\]
and the martingale part by
$M_t := \int_0^t \sigma(s,X_s)dW_s +J_t$.
Suppose for simplicity that $X_0=0$.

The reduction function $r$ is non-decreasing and the non-negative function $\alpha(\cdot,\cdot)$ is bounded by $\pi$ since it satisfies the maximum principle. Thus, for the drift term we have
\begin{equation}\label{ineq}
0\leqslant \int_0^t r(\alpha(s,X_s))ds \leqslant  r(\pi)t.
\end{equation}
Next, we will use Kolmogorov-Doob inequality to bound the martingale part of $\{X_t\}$, for which we need a bound for the second moment of $M_t$.
Observe that the covariance between the diffusion and jump terms is given by
\begin{equation}\label{cross}
\E\left( J_t \int_0^t \sigma(s,X_s)dW_s  \right)  = \E \left[ J_t \E \left(\int_0^t \sigma(s,X_s)dW_s \Big{|} \{T_j, Y_j\} \right) \right],
\end{equation}
where $T_j$ is the time of the $j$-th jump $Y_j$. Considering the inner expectation given $\{T_j=t_j, Y_j=y_j\}$,
one can easily see that it is equal to zero, and hence the expected value (\ref{cross}) is zero as well. It follows that
\[
\E M_t^2 = \E\left( \int_0^t \sigma(s,X_s)dW_s \right)^2 + \E J_t^2.
\]
Clearly $\E (J_t^2) = \lambda\E Y_1^2 t$. Next, assume that $\sigma$ satisfies the standard linear growth condition
\begin{equation*}
\sigma^2(s,x)\leqslant a + bx^2, \qquad a,b\geqslant 0, \qquad x\in\R.
\end{equation*}
Then
\begin{align*}
\E\left( \int_0^t \sigma(s,X_s)dW_s \right)^2 = \int_0^t \E \sigma^2(s,X_s) ds &\leqslant \int_0^t \E\left(a+b X_s^2\right) ds \\
& = a t +  b \int_0^t \E X_s^2 ds,
\end{align*}
and so we have
\[
\E M_t^2 \leqslant \lambda\E Y_1^2 t + a t +  b \int_0^t \E X_s^2 ds.
\]
Since we see from (\ref{ineq}) that
\begin{equation}\label{ineqMX}
M_t - r(\pi)t \leqslant X_t \leqslant M_t
\end{equation}
and $\E M_t =0$, one has
\[
\E X_t^2 \leqslant r^2(\pi)t^2 + \E M_t^2 \leqslant (\lambda\E Y_1^2 + a)t + r^2(\pi)t^2 +  b \int_0^t \E X_s^2 ds.
\]
Hence Gronwall's lemma (see (2.10) and (2.11), Chapter 5 in \cite{KaratzsasShreve1}) applied to $\{X_t\}$ provides the upper bound
\begin{equation}\label{gronwall}
\E X_t^2 \leqslant \frac{1}{b}\left(a+\lambda\E Y_1^2 +\frac{2r^2(\pi)}{b}\right)e^{bt}-\frac{1}{b}\left(a+\lambda\E Y_1^2+2r^2(\pi)t+\frac{2r^2(\pi)}{b}\right)=:\kappa_t.
\end{equation}
The Kolmogorov-Doob inequality applied to the submartingale $\{-X_t\}$ together with (\ref{gronwall}) yields
\begin{equation*}
\Pro \left( \inf_{0\leqslant t \leqslant T} X_t \leqslant -l \right) =
\Pro \left( \sup_{0\leqslant t \leqslant T} (-X_t) \geqslant l\right)
\leqslant \kappa_T l^{-2} =:\epsilon_1
\end{equation*}
for all $l>0$. Noting that
\[
\E M_t^2 \leqslant \kappa_t + r^2(\pi)t^2 =: \zeta_t,
\]
one obtains similarly from the right inequality in (\ref{ineqMX}) that
\begin{equation*}
\Pro \left( \sup_{0\leqslant t \leqslant T} X_t \geqslant l \right) \leqslant
\Pro \left( \sup_{0\leqslant t \leqslant T} M_t \geqslant l\right)
\leqslant \zeta_T l^{-2} =:\epsilon_2,
\end{equation*}
which provides a way to determine the domain of truncation $[-l, l]$ for our problem for given accuracy levels $\epsilon_1$, $\epsilon_2$.

Suppose that the time unit is one year and assume that the time to compliance is $T=1$. Set the penalty level at $\pi=1$
and take  for simplicity a constant function
$\sigma(t, x)\equiv\sigma$ in the diffusion term. We consider $\sigma=1$ here.
The compound Poisson process is realized with intensity rate $\lambda=1$ and a standard normal jump distribution
$\nu=N(0, 1)$. Finally, we suppose that the cumulative reduction function $r$
is linear: $r(a)=a$ for all $a \in [0, \pi]$. In this case, for $\epsilon_1=\epsilon_2=0.05$, we obtain the domain truncation boundary $l\approx 11$.

\begin{figure}
\centering
    \includegraphics[width=10.5cm, height=5cm]{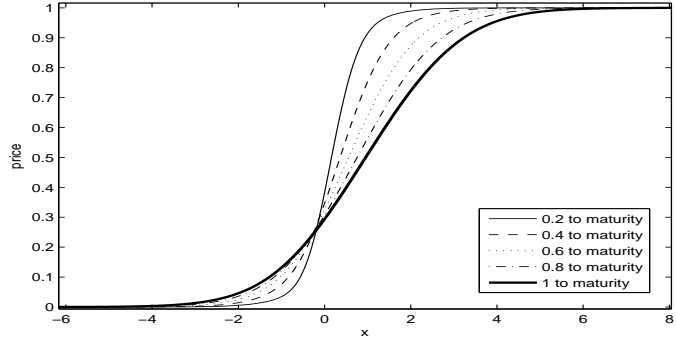}
    \caption{Function $\alpha(t, \cdot)$ for $t=0.2,0.4, 0.6, 0.8, 1$ calculated by discretization (\ref{newproblem}) with parameters given in the text.}
    \label{fig:PriceProcess}
\end{figure}
\begin{figure}
\centering
    \includegraphics[width=10.5cm, height=5cm]{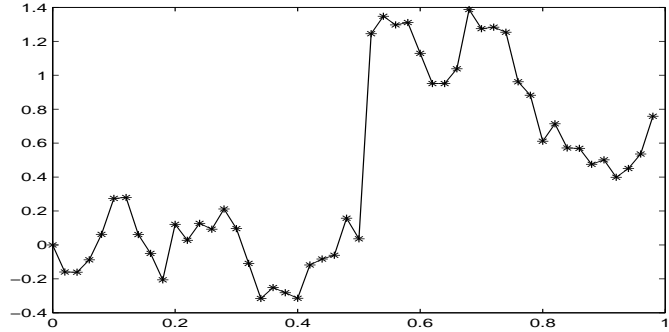}
    \caption{A typical path of the  jump-diffusion process $\{X_{t}\}_{t \in [0, T]}$.}
    \label{fig:PriceProcess1}
\end{figure}
\begin{figure}
\centering
    \includegraphics[width=10.5cm, height=5cm]{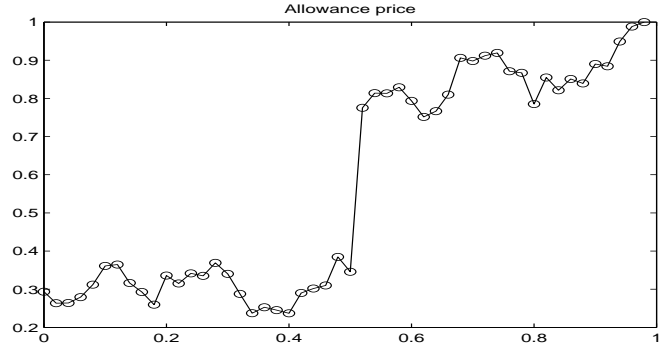}
    \caption{The corresponding realization of the  allowance price process $\{A_t=\alpha(t, X_{t})\}_{t \in [0, T]}$ for the $\{X_{t}\}_{t \in [0, T]}$ obtained in Figure \ref{fig:PriceProcess1} .}
    \label{fig:PriceProcess2}
\end{figure}

We consider a sequence of discretization schemes  with the
time step  $\Delta_{t}=0.02$ which gives $M=50$ time points on the grid.
Choosing the same space discretization $\Delta_{x}=0.02$, we obtain  $N=2l/0.02=100 l$  space points, depending on the truncation $[-l, l]$ of the space region.
To control the error made by the truncation, the $L_1$-norm of the difference between two solutions $\beta_{i}^{n}(l)$,  $\beta_{i}^{n}(\tilde l)$ corresponding to the truncations at levels $\pm l$ and $\pm \tilde l$, respectively, is considered:
\begin{equation*}
\sum_{i}\sum_{n=0}^{M} |\beta_i^n(l) - {\beta}_i^n(\tilde {l})|, \qquad  \label{error}
l< \tilde l,
\end{equation*}
where the summation is taken over all grid points on a fixed  sub-region $[0, T] \times [-d,d]$ with $d<l<\tilde l$. Numerical experiments show that, for $d=11$ and $l=20$,  $\tilde l=30$, the error is of order $10^{-7}$. We therefore conclude that the solution $\{\beta_i^n\}_{i,n}$ calculated for the space region  $[-30,30]$ is accurate.
Figure \ref{fig:PriceProcess}  shows  the shape of the  function $\alpha$  obtained in this way.
The graph in Figure \ref{fig:PriceProcess1} depicts  a typical  realization of the process $\{X_t\}_{t \in [0, T]}$,
calculated using  the forward Euler method \cite{KloeP}. Figure \ref{fig:PriceProcess2}   shows
the corresponding realization of the price process $\{A_t=\alpha(t,X_{t})\}_{t \in [0, T]}$.

Figure \ref{fig:Comparison} displays the function $\alpha$ at time $t=0.8$ obtained for various volatility levels $\sigma$ and jump rates $\lambda$. It turns out that, when the volatility is small, the price process tends faster to its boundary value: agents do not expect big changes in pollution emissions and therefore in the allowance price. If there is a shortage (excess) in allowance credits before maturity, the market expects that there will also be a shortage (excess) at $t=T$. In the presence of jumps or a large diffusion coefficient, the allowance price tends to converge slower as one can expect a sudden increase or decrease in pollution emissions which impacts directly on the allowance price.

\begin{figure}
\centering
    \includegraphics[width=10.5cm, height=5cm]{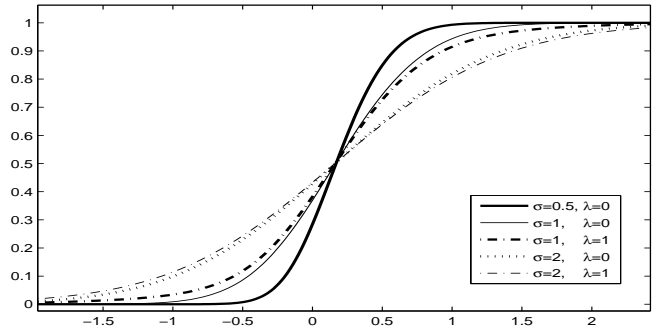}
    \caption{Function $\alpha(0.8, \cdot)$ for different values of $\sigma$ and $\lambda$.}
    \label{fig:Comparison}
\end{figure}

\section{Conclusions}
The growing evidence of the cost of climate change justifies  the introduction
of large-scale  measures. Emission trading schemes become increasingly important.
The  proposed operation of newly designed  cap-and-trade
mechanisms within a  multi-period setting and with inter-connection to other markets
 adds a notable  complexity to such systems,
and so quantitative understanding of emission trading schemes becomes
increasingly challenging. For instance, the problems of  market design,
the emission reduction performance, the optimization of  allowance allocation procedures,
the individual risk management and the valuation of emission-related financial instruments
need to be addressed within a sound mathematical framework, which we aim to
 approach in this work.
 Based on results from equilibrium analysis,  we focus on  the simplest situation of a
 one-period market to show how  the risk-neutral
evolution of emission allowances can be described in terms of jump-diffusion processes.
Although the resulting partial integro differential equations are non-linear,
we provide numerically stable and fast valuation procedures which yield reliable numerical
 schemes for valuation of derivatives of the fast-growing family of emission-related financial assets.

%\bibliographystyle{plain}
%\bibliography{co2}
%\bibliography{C:/Users/juri/Desktop/artikel/EUA_Option_Pricing/co2}

\end{document}